\documentclass[12pt]{article}

\usepackage{subfigure}
\usepackage{amssymb}
\usepackage{amsfonts}
\usepackage{bm}
\usepackage{color}
\usepackage{calc}
\usepackage[latin1]{inputenc}
\usepackage{amsfonts}
\usepackage{amsmath}
\usepackage{amsthm}
\usepackage[english]{babel}
\usepackage[T1]{fontenc}
\usepackage{verbatim}
\usepackage{subfig}
\usepackage{enumerate}
\usepackage{graphicx}
\usepackage{url}
\usepackage{algorithmic,algorithm}
\usepackage{placeins}
\usepackage{trimclip}

\newtheorem{proposition}{Proposition}
\newtheorem{theorem}{Theorem}
\newtheorem{lemma}{Lemma}
\def\Bmp#1{ \begin{minipage}{#1} }
\def\Bmpc#1{ \begin{minipage}[c]{#1} }
\def\Bmpt#1{ \begin{minipage}[t]{#1} }
\def\Bmpb#1{ \begin{minipage}[b]{#1} }
\def\Emp{ \end{minipage} }

\DeclareMathOperator{\arccoth}{arcCoth} 

\def\CC{{\mathbb{C}}}

\def\RR{{\mathbb{R}}}

\def\NN{{\mathbb{N}}}
\def\ZZ{{\mathbb{Z}}}

\def\Dpartial#1#2{ {\frac{\partial #1}{\partial #2} }}

\makeatletter
\providecommand*{\upY}{%
  \mathbin{%
    \mathpalette\@updownY{0}%
  }%
}
\providecommand*{\downY}{%
  \mathbin{%
    \mathpalette\@updownY{1}%
  }%
}
\providecommand*{\UpDownYFactor}{1}
\newcommand*{\@updownY}[2]{%
  \sbox0{$#1+\m@th$}%
  \dimen2=.5\dimexpr\wd0-\ht0-\dp0\relax
  \sbox2{$#1\vcenter{}$}%
  \dimen4=\dimexpr\ht0-\ht2\relax
  \setbox0=\hbox to 0pt{%
    \hss
    \clipbox{%
      0pt %
      {\dimexpr\totalheight-\UpDownYFactor\dimen4\relax} %
      0pt %
      -\dimen2%
    }{$#1|$}%
    \hss
  }%
  \ht0=\dimexpr\ht0-\dimen2\relax
  \kern\dimen2 %
  \raise\ht2\hbox{%
    \ifnum#2=0 %
      {\rotatebox{120}{\copy0}}%
      \copy0 %
      {\rotatebox{-120}{\copy0}}%
    \else
      {\rotatebox{60}{\copy0}}%
      {\rotatebox{180}{\copy0}}%
      {\rotatebox{-60}{\copy0}}%
    \fi
  }%
  \kern\dimen2 %
}
\makeatother

\begin{document}
\title{Rotating Equilibria of Vortex Sheets}

\author{Bartosz Protas$^{1, }$\thanks{Email address for correspondence: bprotas@mcmaster.ca} \ and Takashi Sakajo$^{2}$
\\ \\ 
$^1$ Department of Mathematics and Statistics, McMaster University \\
Hamilton, Ontario, L8S 4K1, Canada
\\ \\
$^2$ Department of Mathematics, Kyoto University \\ 
Kitashirakawa Oiwake-cho, Sakyo-ku, Kyoto, 606-8502, Japan
}

\date{\today}

\maketitle

\begin{abstract}
  We consider relative equilibrium solutions of the two-dimensional
  Euler equations in which the vorticity is concentrated on a union of
  finite-length vortex sheets. Using methods of complex analysis, more
  specifically the theory of the Riemann-Hilbert problem, a general
  approach is proposed to find such equilibria which consists of two
  steps: first, one finds a geometric configuration of vortex sheets
  ensuring that the corresponding circulation density {is
    real-valued and also} vanishes at all sheet endpoints such that
  the induced velocity field is well-defined; then, the circulation
  density is determined by evaluating a certain integral formula. As
  an illustration of this approach, we construct a family of rotating
  equilibria involving different numbers of straight vortex sheets
  rotating about a common center of rotation and with endpoints at the
  vertices of a regular polygon. This equilibrium generalizes the
  well-known solution involving single rotating vortex sheet. With the
  geometry of the configuration specified analytically, the
  corresponding circulation densities are obtained in terms of a
  integral expression which in some cases lends itself to an explicit
  evaluation. It is argued that as the number of sheets in the
  equilibrium configuration increases to infinity, the equilibrium
  converges in a certain distributional sense to a hollow vortex
  bounded by a constant-intensity vortex sheet, which is also a known
  equilibrium solution of the two-dimensional Euler equations.
\end{abstract}

%\begin{keywords}
\begin{flushleft}
\end{flushleft}
%\end{keywords}

\section{Introduction}
\label{sec:intro}

Relative equilibria represent a particularly important class of
solutions of equations governing the motion of fluids. This is because
if they are stable then such equilibria describe flow structures which can
persist for long times. Here we are interested in two-dimensional (2D)
flows of inviscid incompressible fluids with localized vorticity,
although models similar to the ones developed here are also relevant
for plasmas and quantum fluids, in particular, Bose-Einstein
condensates. Arguably, the best-known equilibrium solutions
characterizing such flows involve point vortices where the vorticity
is supported as Dirac measures on a set of discrete points. There
exists a very rich literature on this topic and we refer the reader to
the review papers
\cite{VortexCrystals2003,Aref2007,Clarkson2009,Newton2014} for a
survey. Desingularization of point vortices by replacing them with
finite-area regions of constant vorticity leads to the so-called
vortex-patch problem which is of the free-boundary type, since the
shapes of the patch boundaries need to be found as a part of the
solution of the problem. Such problems have also received significant
attention starting with the seminal work of Pierrehumbert
\cite{Pierrehumbert1980}, Saffman \& Szeto \cite{ss80}, Dritschel
\cite{d85}, Kamm \cite{k87}, Moore et al.~\cite{mst88}, and in this
context we also mention a more recent investigation \cite{fw10a}. 
Some theoretical results concerning the
desingularization process can be found in \cite{Wan1988b} and
\cite{gipz10a}.

Vortex sheets represent discontinuities of the tangential velocity
component, such that vorticity is localized along one-dimensional (1D)
curves \cite{fundam:saffman1}. They have received attention especially
in the context of fluid-structure interaction where they have been
used as simple models of vortex shedding
\cite{Jones2003,hsw07,Shukla2007,Alben2009}. However, unlike the case
of point-vortex or vortex-patch equilibria, our knowledge of
equilibrium configurations involving vortex sheets is rather limited.
Two closed-form solutions are available as regards equilibria involving
closed or unbounded vortex sheet (i.e., sheets possessing no
endpoints): the circular vortex sheet and the infinite (or periodic)
straight vortex sheet, both with constant circulation density and
defined in an unbounded domain
\cite{fundam:batchelor,fundam:saffman1}. Closed vortex sheets play a
role in the Sadovskii flows \cite{s71} in which they separate regions
with constant vorticity from the irrotational flow and they were also
used to construct ``hollow vortices'' in which the vortex sheet
separates a stagnant region of the flow from the region in potential motion
\cite{Pocklington1895,bss76}. More recently such configurations
were obtained for flows past an obstacle \cite{tz11}. A distinguishing
feature of all of these solutions is that closed vortex sheets
separate regions characterized by different values of the Bernoulli
constant.

As regards relative equilibria involving {\em non-closed} vortex
sheets, i.e., flows in which the vorticity support is not a closed
curve, there is one solution known in the analytic form, namely, a
straight sheet of length $2a$ ($a>0$) rotating about its center with
the angular velocity $\omega$ \cite{fundam:batchelor}. The circulation
density is given by
\begin{equation}
\gamma(x) = \omega \sqrt{a^2 - x^2}, \qquad x \in [-a,a],
\label{eq:sp1}
\end{equation}
and, as shown in \cite{fundam:batchelor}, this equilibrium can be
regarded as the limiting solution obtained from Kirchhoff's rotating
ellipse by fixing its major axis and letting the minor axis vanish
{while allowing the vorticity to become unbounded in such a way
  that the circulation remains constant}.  Several approximations of
relative sheet equilibria were computed by O'Neil \cite{ONeil2009},
see also \cite{ONeil2010}, who represented the sheets using a number
of unconnected point vortices arranged along finite arcs. More
recently, an exact method has been devised by the same author to
compute relative equilibrium configurations involving both vortex
sheets and point vortices \cite{ONeil2018a,ONeil2018b}.

In the present investigation we introduce a novel approach allowing
one to construct a family of rotating equilibria involving vortex
sheets which generalizes the equilibrium solution \eqref{eq:sp1}.
Hereafter heavy use will be made of complex notation. Assuming that
$z(\xi,t) = x(\xi,t) + iy(\xi,t)$ represents the position at time $t$
of a point on a vortex sheet $L$ corresponding to the arc-length
parameter $\xi$, the evolution of the sheet is governed by the
Birkhoff-Rott singular integro-differential equation
\cite{fundam:saffman1}
\begin{equation}
\Dpartial{z^\ast}{t}(\xi,t) = V(z(\xi,t))
\,\equiv\, \frac{1}{\pi i}\mbox{pv} \int_L \frac{\gamma(s)}{ z(s,t) - z(\xi,t) }\,ds,
\label{eq:BR}
\end{equation}
where $\gamma(s)$ is the circulation density of the vortex sheet which
can be interpreted as the difference between the tangential velocity
components on the two sides of the sheet and the asterisk denotes
complex conjugation. In \eqref{eq:BR} integration is defined in the
sense of Cauchy's principal value.  We note that in general the sheet
$L$ may consist of a number of disjoint segments. Because of its
ill-posed nature, the Birkhoff-Rott equation \eqref{eq:BR} has
generated a lot of interest in the applied mathematics literature
\cite{c89a,mb02}.  Since we are interested here in relative
equilibria, i.e., equilibria attained in a moving --- rotating or
translating --- frame of reference, the complex velocity $V(z) =
(u-iv)(z)$, where $u$ and $v$ are, respectively, the $x$ and $y$
components, will be given in terms of some prescribed function $f(z)$
independent of time. For relative equilibria rotating with the rate of
rotation $\omega \in \RR$, this function has the form $f(z) = i \omega
z^*$, whereas for relative equilibria translating with a constant
velocity $U \in \CC$, it has the form $f(z) = U$. We note that in
either case the function $f(z)$ is defined up to an arbitrary real
multiplicative constant which can be factored out by rescaling the
circulation density $\gamma(s)$.

The particular question we consider in this study is whether there
exist relative equilibrium solutions of the Birkhoff-Rott equation
\eqref{eq:BR} consisting of a certain number $p \ge 1$ of sheets
$L_m$, $m=0,\dots,p-1$, assumed disjoint, possibly except for a finite
number of points of contact, such that we have $L =
\bigcup_{m=0}^{p-1} L_m$ in \eqref{eq:BR}. The endpoints of the sheets
will be denoted $a_m, b_m \in \CC$ ($a_m, b_m \in L_m$),
$m=0,\dots,p-1$. Thus, the problem of finding such relative equilibria
is equivalent to the question whether the integral equation
\begin{equation}
\frac{1}{\pi i} \mbox{pv} \int_L \frac{\gamma(s) \, ds}{z(s) - z(\xi)} = f(z), \qquad \forall z(\xi) \in L
\label{eq:BR0}
\end{equation}
admits {suitable} solutions. We emphasize that by ``solution'' we
mean here the pair $\{L,\gamma\}$ representing the geometry of the
vortex sheet and the corresponding {real-valued} distribution of
the circulation density.

Using methods of complex analysis, more specifically, the
Riemann-Hilbert theory, we develop a constructive approach allowing
one to obtain relative equilibrium solutions to \eqref{eq:BR0} which
in some situations are given in a closed form. In this formulation the
problems of determining the shapes $L$ of the sheets and the
corresponding circulation densities $\gamma$ decouple: first, one
needs to find the sheets $L$ satisfying certain conditions after which
obtaining the circulation density $\gamma$ reduces to evaluating an
integral formula. {In this sense, there is an analogy between the
  proposed approach and the ``Brownian ratchet'' method developed by
  Newton to find relative equilibria of point vortices
  \cite{Newton2007,bbn10a}.}  Using the proposed approach we then
construct a family of rotating equilibrium configurations consisting
of a different number of sheets each with one endpoint at the center
of rotation and the other at a vertex of a polygon. This family
generalizes the well-known equilibrium solution \eqref{eq:sp1}
describing a single rotating sheet and in a certain ``weak'' sense
approaches a hollow vortex bounded by a closed constant-density vortex
sheet when the number of sheets $p$ increases to infinity.

The structure of the paper is as follows: in the next section we
describe the connection between the questions considered here and the
Riemann-Hilbert problem in complex analysis which allows us to
formulate our approach; then in Section \ref{sec:rot} we introduce a
particular family of configurations of rotating vortex sheets, prove
that they constitute relative equilibria and and obtain expressions for
the circulation densities; next, in Section \ref{sec:numer}, we
describe and validate a numerical approach needed to evaluate these
formulas; the results are presented in Section \ref{sec:results},
whereas discussion and final conclusions are deferred to Section
\ref{sec:final}.

\section{Vortex Sheets and the Riemann-Hilbert Problem}
\label{sec:RH}

In this section we first elucidate the relation between the problem
considered here and the Riemann-Hilbert problem of complex analysis.
Then, some elements of the theory of this problem will be used to
formulate a two-step solution approach. {Given a contour on the
  complex plane,} the Riemann-Hilbert problem consists in finding a
holomorphic function defined in {the complement of this contour
  on the complex plane such that its limiting values on both sides of
  the contour satisfy a prescribed relation \cite{af11}}.  Since
holomorphic functions are conveniently expressed in terms of singular
(Cauchy-type) integrals, there is a natural connection between the
Riemann-Hilbert problem and the theory of singular integral equations.

In a similar spirit, to determine a relative equilibrium involving
vortex sheets we need to find a velocity field, which for potential
(i.e., incompressible and irrotational) flows is represented in terms
of a holomorphic function {$V(z)$, $z \in \CC \backslash L$,
  satisfying the following kinematic conditions on the contour (sheet)
  $L$
\begin{equation}
V^+(t) - V^-(t) = \gamma(t) \frac{dt}{ds}(t), \qquad t \in L,
\label{eq:RH}
\end{equation}
where $t = t(s)$ represents the parameterization of the contour $L$ in
terms of its arc length $s$, whereas $V^\pm(t)$ are the limits of the
velocity field $V(z)$ as $z \rightarrow t \in L$ on the two sides of
the contour. Condition \eqref{eq:RH} implies} that the normal velocity
component on both sides of the sheet {is} the same and the
tangential velocity component {has} a jump equal to the
prescribed circulation density $\gamma$ (the continuity of pressure
across the sheet then follows from these conditions and the Bernoulli
equation).  {For the Riemann-Hilbert problem to be completely
  defined \cite{af11}, relation \eqref{eq:RH} must be complemented by
  a condition specifying the behavior of the function $V(z)$ as $|z|
  \rightarrow \infty$, which is determined by the function $f(z)$.
  These conditions together are equivalent to relation} \eqref{eq:BR0}
and thus the question of the existence of relative equilibria
involving vortex sheets can be recast in terms of whether equation
\eqref{eq:BR0} admits suitable solutions.

Aspects of the Riemann-Hilbert problem relevant to the questions
considered here are reviewed in monograph \cite{Muskhelishvili2008},
from which we have adopted parts of the notation. First, the integral
expressing velocity needs to be recast as a complex integral defined
with respect to the complex variable $t$, i.e., for any $\zeta \in \CC
\, \backslash \, L$ we have
\begin{equation}
V(\zeta) = \frac{1}{2\pi i} \int_L \frac{\gamma(t(s)) \, ds}{t(s) - \zeta} = \frac{1}{2\pi i} \int_L
\frac{\varphi(t) \, dt}{t - \zeta}, \label{eq:V} 
\end{equation}
where the density is defined as 
\begin{equation}
 \quad \varphi(t)  := \gamma(t) \left(\frac{dt}{ds}\right)^{-1}.
\label{eq:varphi}
\end{equation}
Problem \eqref{eq:BR0} can then be recast in a canonical form
consistent with the Riemann-Hilbert problem, namely
\begin{equation}
\frac{1}{\pi i} \mbox{pv} \int_L \frac{\varphi(t) \, dt}{t - z} = f(z), \qquad \forall  \in L.
\label{eq:BR0t}
\end{equation}
Then, we note that in order for the velocity \eqref{eq:V} to be
bounded everywhere in $\CC \, \backslash \, L$, the density
$\varphi(t)$, cf.~\eqref{eq:varphi}, must satisfy certain conditions,
namely, it must be H\"older-continuous with some H\"older index $0 <
\mu \le 1$ and must vanish at all sheet endpoints
\cite{Muskhelishvili2008}, i.e.,
\begin{equation}
\varphi(a_m) = \varphi(b_m) = 0, \quad m=0,\dots,p-1.
\label{eq:gamma0}
\end{equation}
This last condition in particular will play an essential role in our
analysis below. 

In general, problem \eqref{eq:BR0t} admits solutions $\varphi$ which
vanish at $p$ endpoints only \cite{Muskhelishvili2008}.  In order for
the densities $\varphi$ obtained from \eqref{eq:BR0t} to vanish at
{\em all} $2p$ endpoints (which is necessary for the corresponding
velocity field \eqref{eq:V} to be well-defined everywhere), the sheets
$L$ and the function $f(z)$ in \eqref{eq:BR0t} must satisfy the
following set of compatibility condition \cite[relation
(89.14)]{Muskhelishvili2008} (in essence, there is one condition for
each additional endpoint on which $\varphi(t)$ is to vanish)
\begin{equation}
C_n := \int_L \frac{t^n f(t) \, dt}{\sqrt{R_p(t)}} = 0, \quad n=0,\dots,p-1,
\label{eq:Cn}
\end{equation}
where
\begin{equation}
R_p(z) := \prod_{m=1}^p (z-a_m)(z-b_m),
\label{eq:R}
\end{equation}
where ``$:=$'' means ``equal to by definition''.  Since the function
$f(z)$ is fixed and determined by the type of the relative equilibrium
sought (rotating or translating), conditions \eqref{eq:Cn} can be
interpreted as constraints on the shapes of the vortex sheet $L$ which
must be satisfied in order for the corresponding velocity field
\eqref{eq:V} to be well defined.  {In addition, it must also be
  ensured that the resulting circulation density is real-valued which
  can often be done based on some symmetry arguments (cf.~Section
  \ref{sec:rot} below).}  Then, a density satisfying conditions
\eqref{eq:gamma0} is given by the following formula \cite[relation
(88.9)]{Muskhelishvili2008}
\begin{equation}
\varphi(z) = \frac{\sqrt{R_p(z)}}{\pi i}
\int_L \frac{f(t) \, dt}{\sqrt{R_p(t)} (t-z)} = 0, \quad z \in L.
\label{eq:gamma}
\end{equation}
Finally, the real-valued circulation density $\gamma$ needed for the
Birkhoff-Rott equation \eqref{eq:BR0} can be obtained from
\eqref{eq:gamma} and \eqref{eq:varphi} using the known parameterization
of the contour.

The above facts {demonstrate how} the problems of finding the
equilibrium shapes of vortex sheets and the corresponding circulation
densities {are related to each other}. They also suggest the
following general and constructive approach to finding relative
equilibria involving vortex sheets
\begin{enumerate}
\item[Step 1:] find arcs $L_0,\dots,L_{p-1}$ such that the
  compatibility conditions \eqref{eq:Cn} are satisfied {and the
    imaginary part of $\varphi(t) \frac{dt}{ds}(t)$ vanishes},

\item[Step 2:] determine the corresponding circulation densities by
  evaluating formula \eqref{eq:gamma} and using \eqref{eq:varphi}.
\end{enumerate}
In the next section we implement this procedure in order to find a new
family of equilibrium solutions consisting of an arbitrarily large
number of rotating vortex sheets which generalizes the well-known
solution \eqref{eq:sp1}.

%{\color{red}
\section{Rotating Equilibria}
\label{sec:rot}

\begin{figure}
\begin{center}
\includegraphics[width=1.0\textwidth]{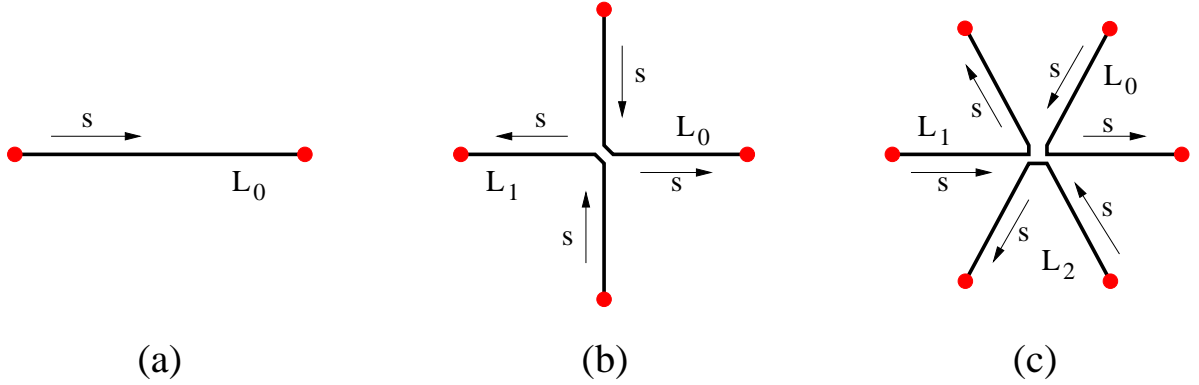}
\caption{Schematic representation of the arrangement of (a) $p=1$, (b)
  $p=2$ and (c) $p=3$ bent sheets into a polygonal configuration. Red
  circles mark the endpoints on which conditions \eqref{eq:gamma0}
  apply, whereas arrows represent the parameterization used in
  \eqref{eq:Lm}.  When $p \ge 2$ midpoints of the sheets come into
  contact at the center of rotation.}
\label{fig:sheets}
\end{center}
\end{figure}

In this section we first consider a particular arrangement of vortex
sheets and demonstrate, through a rigorous proof, that it can form a
rotating equilibrium, cf.~Step 1 in Section \ref{sec:RH}. Then, as
Step 2, we obtain a closed-form integral formula for the circulation
density of the sheet which in the simplest cases lends itself to
explicit evaluation. Since we focus on rotating equilibria, we now set
$f(z) = i \omega z^*$ and consider an arrangement of $p$ bent sheets
with common midpoints, cf.~figure \ref{fig:sheets} (the choice of this
particular configuration as a candidate for a relative equilibrium was
inspired by our earlier computational experiments). {As will be
  shown below, due to their symmetry properties these configurations
  always produce purely-real circulation densities.}  In order to
parameterize this configuration, we first define the $2p$th complex
root of unity as
\begin{equation}
\sigma = \exp\left(\frac{2\pi i}{2p}\right)
\label{eq:sigma}
\end{equation}
with the properties that $\sigma^{2p}=1$ and $|\sigma|=1$. We note
that $\sigma^m$, $m=0,\dots,2p-1$ are the complex coordinates of the
vertices of a regular $2p$-gon inscribed in a unit circle. The sheets
are then described as, cf.~figure \ref{fig:sheets},
\begin{equation}
L_m \ : \ z_m(s) = \left\{ \begin{array}{rr} -\sigma^{2m+1} s, & -1 \leqq s \leqq 0, \\ \sigma^{2m}s, & 0 \leqq s \leqq 1,  \end{array}\right. 
\label{eq:Lm}
\end{equation}
for $m=0, 1, \dots, p-1$ with
\begin{equation*}
z'_m(s) = \left\{ \begin{array}{rr} -\sigma^{2m+1}, & -1 \leqq s \leqq 0, \\ \sigma^{2m}, & 0 \leqq s \leqq 1.  \end{array}\right.,
\end{equation*}
where the prime denotes differentiation with respect to the arc length
parameter $s \in [-1,1]$. The endpoints of the $m$th sheet $L_m$ then
are $a_m=z_m(-1)=\sigma^{2m+1}$ and $b_m = z_m(1) = \sigma^{2m}$, such
that the function $R_p(z)$, cf.~\eqref{eq:R}, can be expressed as
\begin{equation}
R_p(z) = \prod_{m=0}^{p-1} (z-\sigma^{2m})(z-\sigma^{2m+1}) = \prod_{m=0}^{2p-1}(z-\sigma^m).
\label{eq:R2}
\end{equation}
We observe that the function $R_p(z)$ is invariant under
multiplication of its argument by $\sigma^\ell$ for any $\ell \in \NN$
\begin{equation}
R_p(\sigma^\ell z) = \prod_{m=0}^{2p-1}(\sigma^\ell z-\sigma^m) = \sigma^{2p\ell}\prod_{m=0}^{2p-1}(z-\sigma^{m-\ell}) = \prod_{m=0}^{2p-1}(z-\sigma^{m-\ell}) = R_p(z)
\label{eq:Rl}
\end{equation}
because this operation is equivalent to rotating the $2p$-gon formed by
the sheet endpoints by the angle $\ell \pi / p$, a transformation which
leaves this polygon invariant.

We are now in the position to state the following proposition
\begin{proposition}
  Assuming the vortex sheets $L_m$, $m=0,\dots,p-1$ are defined by
  expression \eqref{eq:Lm} and that $f(z) = i \omega z^*$, the
  compatibility conditions \eqref{eq:Cn} are satisfied for any
  positive integer $p$.
\end{proposition}
\begin{proof}
  We begin by observing that, because of property \eqref{eq:Rl}, the
  denominators in the integrand expressions in constraints
  \eqref{eq:Cn} remain unchanged when the constraint expressions are
  evaluated for different sheets. Thus, noting that
\begin{equation*}
C_n = i \omega \int_L \frac{t^n t^* \, dt}{\sqrt{R_p(t)}} 
= i \omega \sum_{m=0}^{p-1}\int_{L_m} \frac{t^n t^* \, dt}{\sqrt{R_p(t)}}, \quad n=0,\dots,p-1,
\end{equation*}
we obtain for the integral defined on the $m$th sheet when $n=0$
\begin{align}
i\omega\int_{L_m} \frac{t^*}{\sqrt{R_p(t)}}\,dt  &= i\omega \int_{-1}^0 \frac{(-\sigma^\ast)^{2m+1}s}{\sqrt{R_p(t(s))}} (-\sigma)^{2m+1} \, ds + i\omega \int_0^1 \frac{(\sigma^\ast)^{2m} s}{\sqrt{R_p(t(s))}} \sigma^{2m} \, ds \nonumber \\
       &= i\omega \int_{-1}^0 \frac{s}{\sqrt{R_p(t(s))}} \,ds + i\omega \int_0^1 \frac{s}{\sqrt{R_p(t(s))}} \, ds = 0, \quad m = 0,\dots,p-1,
\label{eq:C_00}
\end{align}
which is true because the expression in the denominator,
cf.~\eqref{eq:R}, remains the same for all sheets $L_m$,
cf.~\eqref{eq:Rl}.  Hence, we have $C_0=0$. Then, when $1 \leqq n <
p$, owing to the identity $z^kz^* dz = z^{k-1}| z |^2 dz$, the
integral corresponding to the $m$th sheet becomes, again recognizing
property \eqref{eq:Rl},
\begin{align*}
i\omega\int_{L_m} \frac{t^n t^*}{\sqrt{R_p(t)}}\,dt &= i \omega \int_{-1}^0 \frac{(-\sigma^{2m+1})^{n-1} s^{n+1}}{\sqrt{R_p(t(s))}} (-\sigma^{2m+1}) \,ds +i \omega  \int_0^1 \frac{(\sigma^{2m})^{n-1} s^{n+1}}{\sqrt{R_p(t(s))}} \sigma^{2m} \, ds\\
&= i \omega \int_{-1}^0 \frac{(-1)^n \sigma^{(2m+1)n} s^{n+1}}{\sqrt{R_p(t(s))}} \,ds +i \omega  \int_0^1 \frac{\sigma^{2mn} s^{n+1}}{\sqrt{R_p(t(s))}}\, ds.
\end{align*}
Adding up the integrals corresponding to the sheets $L_m$,
$m=0,\dots,p-1$, we obtain 
\begin{align}
C_n &=  i \omega (-1)^n \sum_{m=0}^{p-1} \sigma^{(2m+1)n}\int_{-1}^0 \frac{ s^{n+1}}{\sqrt{R_p(t(s))}}\, ds +i \omega  \sum_{m=0}^{p-1} \sigma^{2mn} \int_0^1 \frac{ s^{n+1}}{\sqrt{R_p(t(s))}} \, ds \nonumber \\
& = i \omega \left[ -\sum_{m=0}^{p-1} \sigma^{(2m+1)n} + \sum_{m=0}^{p-1} \sigma^{2mn} \right] \int_0^1 \frac{s^{n+1} \, ds}{\sqrt{R_p(t(s))}} \nonumber \\
& = i \omega \left[ \sum_{m=0}^{p-1} \sigma^{2mn} \right] \left(1 - \sigma^n\right) \int_0^1 \frac{s^{n+1} \, ds}{\sqrt{R_p(t(s))}} = 0, \quad n=1,\dots,p-1, \label{eq:C_k0}
\end{align}
which vanishes because due to the properties of the geometric
progression and the fact that $\sigma^{2n} \neq 1$ for $n=1, \dots, p-1$
we have
\begin{equation*}
\sum_{m=0}^{p-1} \sigma^{2mn} = \frac{1-\sigma^{2pn}}{1-\sigma^{2n}} = 0.
\end{equation*}
Thus, relations \eqref{eq:C_00} and \eqref{eq:C_k0} prove our claim.
%\QED
\end{proof}

Having established that the configuration of $p$ bent vortex sheets
defined in \eqref{eq:Lm}, cf.~figure \ref{fig:sheets}, can give rise
to a rotating equilibrium, we now proceed to determine the density
$\varphi_p(z)$ using formula \eqref{eq:gamma} which we transform as
follows in order to reduce it to a definite integral on $[0,1]$
\begin{align}
\varphi_p(z) 
&= \frac{\sqrt{R_p(z)}}{\pi i} \int_L \frac{i \omega t^* \, dt}{\sqrt{R_p(t)}(t-z)}= \frac{\omega}{\pi} \sqrt{R_p(z)} \int_L \frac{t^* \, dt}{\sqrt{R_p(t)}(t-z)}  \nonumber \\
&= \frac{\omega}{\pi} \sqrt{R_p(z)} \sum_{m=0}^{p-1} \int_{L_m} \frac{t^* \, dt}{\sqrt{R_p(t)}(t-z)} \nonumber \\
&= \frac{\omega}{\pi} \sqrt{R_p(z)} \sum_{m=0}^{p-1} \left[ \int_{-1}^0 \frac{(-\sigma^{2m+1})^* s (-\sigma^{2m+1}) \,ds}{\sqrt{R_p(-\sigma^{2m+1}s)}(-\sigma^{2m+1}s-z)}  + \int_0^1\frac{(\sigma^{2m})^* s (\sigma^{2m}) \, ds}{\sqrt{R_p(\sigma^{2m} s)}(\sigma^{2m}s-z)} \right] \nonumber \\
&= \frac{\omega}{\pi} \sqrt{R_p(z)}\sum_{m=0}^{p-1} \left[ \int_{-1}^0 \frac{-s \, ds}{\sqrt{R_p(s)}(\sigma^{2m+1}s+z)}  + \int_0^1\frac{s \,ds}{\sqrt{R_p(s)}(\sigma^{2m}s-z)} \right] \label{eq:gammapA}\\
&= \frac{\omega}{\pi}\sqrt{R_p(z)} \sum_{m=0}^{p-1} \int_0^1 \left[ \frac{-s}{\sqrt{R_p(s)}(\sigma^{2m+1}s-z)} +\frac{s}{\sqrt{R_p(s)}(\sigma^{2m}s-z)} \right]\, ds \nonumber \\
&= \frac{\omega}{\pi} \sqrt{R_p(z)} \int_0^1 \frac{s}{\sqrt{R_p(s)}}  \sum_{m=0}^{p-1}\left( \frac{-1}{\sigma^{2m+1}s-z} +\frac{1}{\sigma^{2m}s-z} \right) \, ds. \label{eq:gammapB}
\end{align}
Before obtaining a simple formula for the circulation density valid
for any $p \ge 1$, we first derive its particular instances for
$p=1,2,3$.  When $p=1$, cf.~figure \ref{fig:sheets}(a), we have
$\sigma = \exp(i\pi)=-1$ and $R_1(z)=(z-1)(z+1)=z^2-1$. Then, from
\eqref{eq:gammapA} we deduce
\begin{align}
\varphi_1(z) & = \frac{\omega}{\pi} \sqrt{R_1(z)} \left[ \int_{-1}^0 \frac{-s \, ds}{\sqrt{R_1(s)}(-s+z)}  + \int_0^1\frac{s \,ds}{\sqrt{R_1(s)}(s-z)} \right] \nonumber \\
& =  \frac{\omega}{\pi} \sqrt{R_1(z)} \int_{-1}^1 \frac{s \, ds}{\sqrt{R_1(s)}(s-z)}.
\label{eq:gammap1}
\end{align}
At the same time, formula \eqref{eq:gammapB} also gives rise to
\begin{equation}
\varphi_1(z) = \frac{\omega}{\pi}\sqrt{R_1(z)} \int_0^1 \frac{s}{\sqrt{R_1(s)}} \left( \frac{-1}{-s-z} +\frac{1}{s-z}  \right) \,ds = \frac{\omega}{\pi} \sqrt{R_1(z)} \int_0^1 \frac{2s^2 \,ds}{\sqrt{R_1(s)}(s^2-z^2)}
\label{eq:gammap1B}
\end{equation}
which is equivalent to \eqref{eq:gammap1}.

Next, for $p=2$, cf.~figure \ref{fig:sheets}(b), substituting
$\sigma=\exp(i\pi/2)=i$, $\sigma^2=-1$, $\sigma^3=-\sigma=-i$ and
$R_p(z)=R_2(z)=z^4-1$ into \eqref{eq:gammapB}, we obtain
\begin{align}
\varphi_2(z) &= \frac{\omega}{\pi} \sqrt{R_2(z)} \int_0^1 \frac{s}{\sqrt{R_2(s)}} \left( \frac{-1}{\sigma s-z} +\frac{1}{s-z} + \frac{-1}{-\sigma s-z} +\frac{1}{-s-z} \right) \, ds \nonumber \\
&= \frac{\omega}{\pi} \sqrt{R_2(z)} \int_0^1 \frac{s}{\sqrt{R_2(s)}} \left( \frac{2z}{s^2+z^2}  +\frac{2z}{s^2-z^2} \right)\,ds \nonumber \\
&= \frac{\omega}{\pi} \sqrt{R_2(z)} \int_0^1 \frac{4zs^3}{\sqrt{R_2(s)}(s^4-z^4)} \,ds. \label{eq:gammap2}
\end{align}
Then, for $p=3$, cf.~figure \ref{fig:sheets}(c), using
$\sigma=\exp(i\pi/3)$ we obtain the relations $\sigma^3=-1$,
$\sigma^4=-\sigma=(\sigma^2)^\ast=\sigma^{-2}$ and
$\sigma^5=\sigma^\ast=-\sigma^2=\sigma^{-1}$, such that we have
\begin{equation}
R_3(z) = (z-1)(z-\sigma)(z-\sigma^2)(z-\sigma^3)(z-\sigma^4)(z-\sigma^5) = z^6-1.
\label{eq:Rp3}
\end{equation}
Rewriting \eqref{eq:gammapB} as $\varphi_p(z) = \frac{\omega}{\pi}
\sqrt{R_p(z)} \int_0^1 \frac{s}{\sqrt{R_p(s)}} F_p(s,z) \, ds$ we
obtain the following expression for the factor $F_p(s,z)$ when $p=3$
\begin{align}
F_3(s, z) & = \sum_{m=0}^{2}\left( \frac{-1}{\sigma^{2m+1}s-z} +\frac{1}{\sigma^{2m}s-z} \right) \nonumber \\ 
& = 2s \left[ \frac{1}{s^2-z^2} - \frac{s^2-z^2}{s^4+z^2s^2+z^4} \right] 
=\frac{6z^2s^3}{s^6-z^6},
\label{eq:F3}
\end{align}
so that in this case relation \eqref{eq:gammapB} becomes
\begin{equation}
 \varphi_3(z)=\frac{\omega}{\pi}\sqrt{R_3(z)} \int_0^1 \frac{6z^2s^4}{\sqrt{R_3(s)}(s^6-z^6)} \,ds.
\label{eq:gammap3}
\end{equation}
{In order to extrapolate these calculations to an arbitrary
  number $p$ of sheets we use the following lemma
\begin{lemma} 
  For $ 0 < p \in \ZZ$, $0 \neq s \in \RR$ and
  $\sigma=\mbox{e}^{i\frac{\pi}{p}}$, the following equality holds.
\begin{equation}
\sum_{k=0}^{2p-1} \frac{(-1)^k}{z - \sigma^k s} = 
\frac{2ps^pz^{p-1}}{z^{2p}-s^{2p}}, \quad z \in \mathbb{C}. \label{eq:sum}
\end{equation}
\end{lemma}
\begin{proof}
The function
\[
G(z)= \frac{2ps^pz^{p-1}}{z^{2p}-s^{2p}}
\]
has poles at $z=\sigma^ks$, $k=0,\dots, 2p-1$, and the residues
$\alpha_k$ around these poles are given by
\[
\alpha_k = \lim_{z \rightarrow \sigma^k s} (z-\sigma^k s)G(z) 
= \mbox{e}^{-i\pi k} = (-1)^k, \quad k=0,\dots, 2p-1.
\]
Since the function
\[
G(z)- \sum_{k=0}^{2p-1} \frac{(-1)^k}{z - \sigma^k s}, \quad z \in \CC
\]
is entire on the whole complex plane, it must be constant owing to
Liouville's theorem.  Moreover, since this function vanishes as $z
\rightarrow \infty$, this constant is zero, which completes the proof.
\end{proof}
Using relation \eqref{eq:sum} in \eqref{eq:gammapB}, we then obtain}
\begin{equation}
\varphi_p(z)=\frac{\omega}{\pi}\sqrt{1-z^{2p}} 
\int_0^1 \frac{2p z^{p-1}s^{p+1}}{\sqrt{1-s^{2p}}(s^{2p}-z^{2p})} \,ds,
\quad z \in [0,1].
\label{eq:gammapp}
\end{equation}
We remark that since the integrand expression in \eqref{eq:gammapp} is
purely real and the integral is over the segment $[0,1]$ of the real
axis, the quantity defined by \eqref{eq:gammapp} is purely real.  The
circulation density $\gamma(s)$ appearing in the Birkhoff-Rott
equation \eqref{eq:BR0} can then be deduced from \eqref{eq:gammapp}
using relation \eqref{eq:varphi} {and will also be purely real,
  as needed}. Noting that since on the contour $[0,1]$ we have $t=s$,
we can conclude that $\gamma_p(s) = \varphi_p(s)$, i.e., expression
\eqref{eq:gammapp} can be interpreted as the circulation densities for
different values of $p$.  {Consequently, we obtain the following main
  theorem.
\begin{theorem}
  The configuration of $p$ bent vortex sheets \eqref{eq:Lm} with the
  circulation density \eqref{eq:gammapp} is a relative equilibrium
  rotating with a constant rate of rotation $\omega$ around the
  origin.
\end{theorem}}

Since the rate of rotation $\omega$ is arbitrary, without loss of
generality hereafter we set $\omega = 1$.  The integral formula
\eqref{eq:gammapp} can in some cases be evaluated explicitly.  For
example, using the symbolic software {\tt Maple}, we obtain the
following explicit expressions for $\gamma_1(s)$ and $\gamma_2(s)$, $s
\in [0, 1]$,
\begin{subequations}
\label{eq:gammap12}
\begin{align}
\gamma_1(s) &= \sqrt{1-s^2}, \label{eq:gamma1ex} \\
\gamma_2(s) &= 2s \arccoth\left(\frac{1}{\sqrt{1-s^4}}\right), \label{eq:gamma2ex} 
\end{align}
\end{subequations}
where we note that \eqref{eq:gamma1ex} is equivalent to the well-known
formula \eqref{eq:sp1} with $\omega = 1$ and $a = 1$. On the other
hand, formula \eqref{eq:gamma2ex} appears to be a new solution.
Evaluation of formula \eqref{eq:gammapp} for $p \geqq 3$ requires
numerical approximation and our approach is described in the next
section.

\section{Numerical Approach}
\label{sec:numer}

The integral in formula \eqref{eq:gammapp} involves singularities at
$s = z$ and at $s = 1$, hence its numerical evaluation requires
special care and in this section we first introduce our numerical
approach which is then validated against the exact expression
\eqref{eq:gamma2ex} corresponding to the case $p = 2$. While integral
formula \eqref{eq:gammapp} is elegant due to its compactness, it does
not directly lend itself to an accurate approximation with a
quadrature. We thus proceed by isolating the terms representing the
singularity at $s = z$, so that the corresponding Cauchy-type
integrals can be evaluated explicitly. The integrals involving the
remaining terms are then approximated with a standard quadrature after
the singularity at $s = 1$ is eliminated with a suitable change of
variables.

For convenience, when evaluating the singular part we will
simultaneously consider contributions from the contours $[-1,0]$ and
$[0,1]$ (both of which are present in the equilibrium configuration
$L$ for all values of $p$, cf.~figure \ref{fig:sheets}). Defining $I_p
:= \int_0^1 \frac{2p z^{p-1}s^{p+1}}{\sqrt{-R(s)}(s^{2p}-z^{2p})}
\,ds$, such that $\varphi_p(z)=\frac{\omega}{\pi}\sqrt{-R_p(z)} \,
I_p(z)$ in \eqref{eq:gammapp}, and
\begin{equation}
G_p(s, z) = \sum_{\substack{m=1 \\ m \neq p}}^{2p-1} \frac{(-1)^m}{\sigma^m s - z},
\label{eq:Gp}
\end{equation}
we obtain $I_p(z) = I^s_p(z) + I^r_p(z)$, where the singular and
regular parts, respectively, $I^s_p(z)$ and $I^r_p(z)$, are
represented as follows for $p=1,2,\dots$
\begin{equation}
I_p^s(z) = \left\{
\begin{alignedat}{2}
& \int_0^1 \frac{2zs \, ds}{\sqrt{-R_p(s)}(s^2 - z^2)}, & \qquad & p \ \text{even} \\
& \int_0^1 \frac{2s^2 \, ds}{\sqrt{-R_p(s)}(s^2 - z^2)}, & \qquad & p \ \text{odd}
\end{alignedat}\right. ,
\label{eq:Ips}
\end{equation}
and
\begin{equation}
I_p^r(z) = \int_0^1 \frac{s}{\sqrt{-R_p(s)}} G_p(s,z) \, ds.
\label{eq:Ipr}
\end{equation}
We remark here that $R_p(z) = z^{2p} - 1 \le 0$ for any $z \in [0,1]$.
Hence, the ``-'' sign preceding $R_p(z)$ ensures that the square root
is real-valued.

\begin{figure}
\begin{center}
\includegraphics[width=0.7\textwidth]{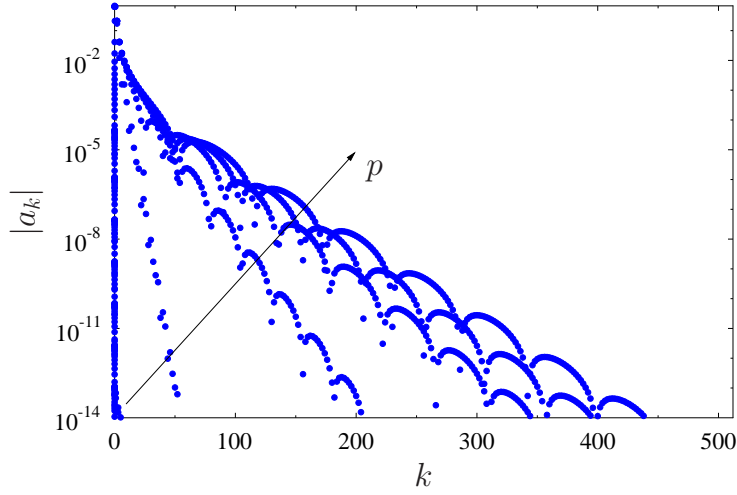}
\caption{Dependence of the magnitude of the Fourier coefficients $a_k$
  of the function $f_p(\theta)$, cf.~\eqref{eq:fp_spec}, on the
  wavenumber $k$ for different values of $p = 10, 109, 208, \dots,
  1000$. The trend with the increase of $p$ is indicated by an arrow.}
\label{fig:fp_spec}
\end{center}
\end{figure}

We first consider evaluation of the singular part $I_p^s(z)$,
cf.~\eqref{eq:Ips}, when $p$ is even. {We introduce a new
  function $r_p(z) := R_p(z) / (z^2 - 1)$ and in order to eliminate
  the singularity at $s = 1$ use the substitutions $s = \cos\theta$
  and $z = \cos\eta$ to change variables from $s$ and $z$ to
  $\theta,\eta \in [0,\pi/2]$.} Then, the singular part for even $p$
becomes
\begin{equation}
I_p^s(\cos\eta) = - 2 \cos\eta \, \int_{\pi/2}^0 \frac{\cos\theta \, d\theta}{\sqrt{r_p(\cos\theta)} \left(\cos^2\theta - \cos^2\eta\right)}.
\label{eq:Ips2}
\end{equation}
Since it does not appear possible to evaluate the integral in
\eqref{eq:Ips2} in closed form for values of $p > 2$, we expand the
function $f_p(\theta) := \frac{1}{\sqrt{r_p(\cos\theta)}}$ in a
Fourier series which, due to the fact that $f_p(\theta)$ is an even
function of $\theta$, reduces to a cosine series
\begin{equation}
f_p(\theta) = \sum_{k=-\infty}^{\infty} a_k e^{i k \theta} = a_0 + 2 \sum_{k=1}^{\infty} a_k \cos(k \theta), \quad a_k \in \RR.
\label{eq:fp_spec}
\end{equation}
As shown in figure \ref{fig:fp_spec}, the magnitudes of $|a_k|$,
$k=0,1,\dots$, of the Fourier coefficients in \eqref{eq:fp_spec} decay
exponentially, although the rate of decay is slower for larger values
of $p$.  In addition, the Fourier coefficients are nonzero for even
wavenumbers only, i.e., $a_{2k-1} = 0$, $k=1,2,\dots$. Using
representation \eqref{eq:fp_spec}, the singular integral
\eqref{eq:Ips2} then becomes
\begin{equation}
I_p^s(\cos\eta) = - 2 \cos\eta \bigg[ a_0 \underbrace{\int_{\pi/2}^0 \frac{\cos\theta \, d\theta}{\cos^2\theta - \cos^2\eta}}_{\Psi_0(\eta)}
+ 2 \sum_{k=1}^{\infty} a_k \underbrace{\int_{\pi/2}^0 \frac{\cos(k\theta) \cos\theta \, d\theta}{\cos^2\theta - \cos^2\eta}}_{\Psi_k(\eta)} \bigg], 
\quad p \ \text{even}.
\label{eq:Ipse}
\end{equation}
Following analogous steps to transform the singular integral
\eqref{eq:Ips} when $p$ is odd, we obtain
\begin{equation}
I_p^s(\cos\eta) = {- 2\,} a_0 \underbrace{\int_{\pi/2}^0 \frac{\cos^2\theta \, d\theta}{\cos^2\theta - \cos^2\eta}}_{\Phi_0(\eta)}
{- 4} \sum_{k=1}^{\infty} a_k \underbrace{\int_{\pi/2}^0 \frac{\cos(k\theta) \cos^2\theta \, d\theta}{\cos^2\theta - \cos^2\eta}}_{\Phi_k(\eta)}, 
\quad p \ \text{odd}.
\label{eq:Ipso}
\end{equation}
In numerical computations the sums in \eqref{eq:Ipse} and
\eqref{eq:Ipso} are truncated after $M$ terms and the value of $M$ can
be deduced from the data in figure \ref{fig:fp_spec} such that this
truncation results in errors not exceeding the machine precision.  The
singular integrals
\begin{subequations}
\label{eq:PsiPhi} 
\begin{align}
\Psi_k(\cos\eta) & := \int_{\pi/2}^0 \frac{\cos(k\theta) \cos\theta \, d\theta}{\cos^2\theta - \cos^2\eta},  \label{eq:Psik} \\
\Phi_k(\cos\eta) & := \int_{\pi/2}^0 \frac{\cos(k\theta) \cos^2\theta \, d\theta}{\cos^2\theta - \cos^2\eta} \label{eq:Phik} 
\end{align}
\end{subequations}
can be evaluated for $k=0,1,\dots$ in closed form, for example using
the symbolic software {\tt Maple}, for different values of the
wavenumber $k$ and for larger values of $k$ this is facilitated by the
recursive relations they satisfy. See \ref{sec:PsiPhi} for details.

The regular part $I_p^r(z)$, cf.~\eqref{eq:Ipr}, is treated by first
transforming it to the variables $\theta,\eta \in [0,\pi/2]$ via the
substitutions $s = \cos\theta$ and $z = \cos\eta$, which removes the
singularity at $s = 1$ and yields
\begin{equation}
I_p^r(\cos\eta) = - \int_{\pi/2}^0 \frac{\cos\theta}{\sqrt{r_p(\cos\theta)}} G_p(\cos\theta,\cos\eta), d\theta.
\label{eq:Ipr2}
\end{equation}
Then, integral \eqref{eq:Ipr2} can be approximated with a standard
approach such as the Gauss-Legendre quadrature.

We remark that given the presence of the factor $1/\sqrt{1-s^2}$,
which is the Chebyshev weight, in integrals
\eqref{eq:Ips}--\eqref{eq:Ipr}, one may be tempted to transform these
integrals to the interval $[-1,1]$ and then approximate them using
Chebyshev quadratures. However, the extension of the real-valued
function $\frac{1}{\sqrt{r_p(s)}}$ to the complex plane involves
branch cuts, such that these integrals become multi-valued when one
integrates across zero. Thus, it is important for the integrals
\eqref{eq:Ips}--\eqref{eq:Ipr} to be defined on $[0,1]$ where the
integrand expressions remain single-valued.

\begin{figure}
\begin{center}
\includegraphics[width=0.7\textwidth]{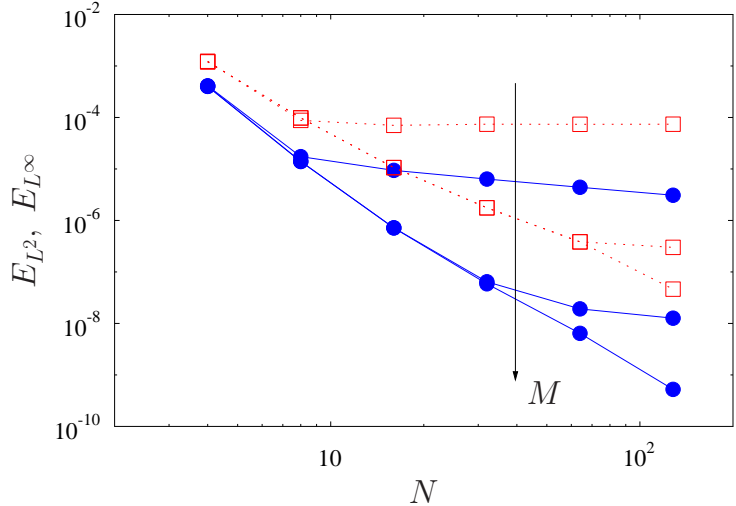}
\caption{Dependence of (blue solid circles) the $L^2$ error
  \eqref{eq:E2} and (red empty squares) $L^{\infty}$ error
  \eqref{eq:Einf} on the number of grid points $N$ for different
  truncation levels $M = 10, 15, 20$. The trend with the increase of
  $M$ is indicated by an arrow.}
\label{fig:ErrM}
\end{center}
\end{figure}
We now proceed to validate the approach described above and do so by
comparing the numerical approximation of \eqref{eq:gammapp} for $p =
2$ with the exact solution \eqref{eq:gamma2ex}.  There are two
numerical parameters in the problem, namely, the number of terms $M$
after which the Fourier series in \eqref{eq:Ipse} is truncated and the
number of gridpoints $N$ in the Gauss-Legendre quadrature used to
approximate \eqref{eq:Ipr2}. Denoting the solution obtained
numerically with these parameters $\gamma^{M,N}_2$, we consider the
following two error measures
\begin{subequations}
\label{eq:E} 
\begin{align}
E_{L^2} & := \| \gamma^{M,N}_2 - \gamma_2 \|_{L^2(0,1)}, \label{eq:E2} \\
E_{L^\infty} & := \| \gamma^{M,N}_2 - \gamma_2 \|_{L^{\infty}(0,1)}. \label{eq:Einf} 
\end{align}
\end{subequations}
The dependence of these errors on the number of grid points $N$ for
different truncation levels $M$ is shown in figure \ref{fig:ErrM},
where we see that provided $M$ is sufficiently large, the error
measures \eqref{eq:E2} and \eqref{eq:Einf} show, respectively,
fourth-order and third-order rate of convergence with respect to grid
refinement in the quadrature (spectral convergence is not achieved
because of the presence of the square-root function in the integrand
of \eqref{eq:Ipr2}). In the computational results presented in the
next section we use $N = 2048$ (which gives a much finer resolution
than used in figure \ref{fig:ErrM}), whereas the parameter $M$ is
adjusted for different $p$ based on the data presented in figure
\ref{fig:fp_spec} to ensure that errors due to truncation of the
series in \eqref{eq:Ipse} remain at the level of the machine
precision.

\section{Results}
\label{sec:results}

In this section we present circulation densities obtained for
different values of $p$ by approximating formula \eqref{eq:gammapp} as
discussed in Section \ref{sec:numer}. Since the circulation densities
are defined up to a multiplicative factor $\omega$, we normalize them
such that the equilibrium configurations $L$ corresponding to
different values of $p$, cf.~\eqref{eq:Lm}, have the same {\em total}
circulation equal to one. We will thus focus on the following
normalized quantity
\begin{equation}
\widetilde{\gamma}_p(s)  := \frac{\gamma_p(s)}{2 p \int_0^1 \gamma_p(s')\, ds'}, \quad s \in [0,1], \ p=1,2,\dots
\label{eq:tildegammap}
\end{equation}
which is plotted for $p=1,2,\dots,10$ in figure \ref{fig:gammap}(a).
We recognize that the distributions obtained for $p=1$ and $p=2$
correspond to the exact solutions \eqref{eq:gamma1ex} and
\eqref{eq:gamma2ex}. It is interesting to note that, except for the
case $p=1$, all circulation densities vanish at $s = 0$ which
corresponds to the center of rotation, a property which is not a
priori guaranteed by the general inversion formula \eqref{eq:gamma},
but is evident in the form of \eqref{eq:gammapp}. We also observe that
as $p$ increases the magnitude of $\widetilde{\gamma}_p(s)$ is reduced
which is a consequence of the fact that since the circulation of the
entire equilibrium configuration is fixed, the circulations of the
individual sheets decrease as their number $p$ goes up. In addition,
in figure \ref{fig:gammap}(a) we observe that there is no apparent
difference between the distributions $\widetilde{\gamma}_p(s)$
obtained for $p$ even and $p$ odd, and that these distributions become
more skewed towards the outer endpoint of the sheet at $s = 1$ as $p$
increases.  This latter effect is further explored in figure
\ref{fig:gammap}(b) where we show the normalized circulation densities
$\widetilde{\gamma}_p(s)$ corresponding to much higher values of $p =
20, 119,\dots, 2000$. In this figure we observe that
$\widetilde{\gamma}_p(s) \rightarrow 0$ as $p \rightarrow \infty$ for
all values of $s \in [0,1-\epsilon)$, $\epsilon > 0$, except for a
small neighborhood near the endpoint $s = 1$ in which the entire
circulation of the vortex sheet becomes localized in the form of a few
sign-changing oscillations. In order to shed further light on this
behavior, in figure \ref{fig:gammapZ} we show the quantity $2p \,
\widetilde{\gamma}_p(s)$ in a small neighborhood of the right endpoint
$s = 1$, corresponding to $0.625 \%$ of the entire interval $[0,1]$.
We observe that in the limit $p \rightarrow \infty$ the dominant
``wiggle'' becomes steeper and more localized. The significance of
this observation is discussed in the next, final section.

\begin{figure}
\centering \vspace*{-3.0cm}
\subfigure[]{\includegraphics[width=0.7\textwidth]{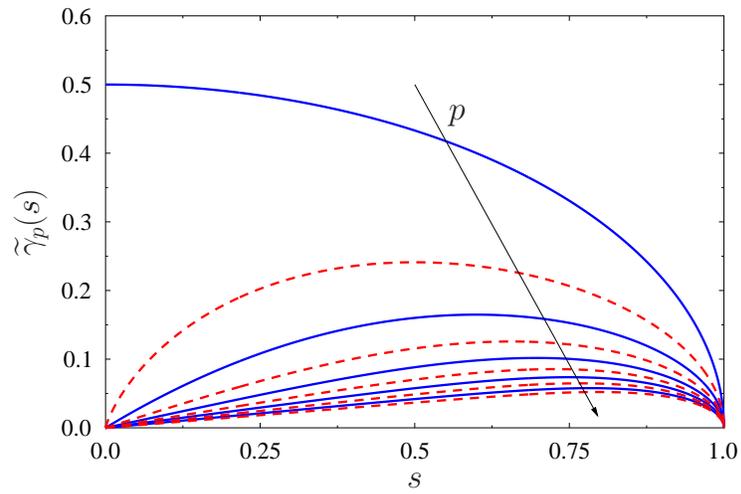} } 
\subfigure[]{\includegraphics[width=0.7\textwidth]{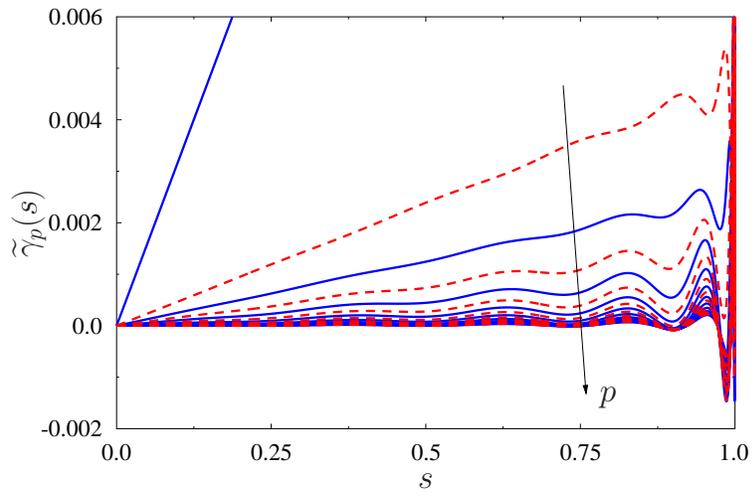}}
\caption{Normalized circulation densities $\widetilde{\gamma}_p(s)$,
  cf.~\eqref{eq:tildegammap}, as functions of position $s \in [0,1]$
  for (a) $p=1,2,\dots,10$ and (b) $p=20,119,\dots,2000$. Red dashed
  lines and blue solid lines correspond to, respectively, $p$ even and
  odd. The trend with the increase of $p$ is indicated by an arrow.}
\label{fig:gammap}
\end{figure}

\begin{figure}
\centering 
\includegraphics[width=0.7\textwidth]{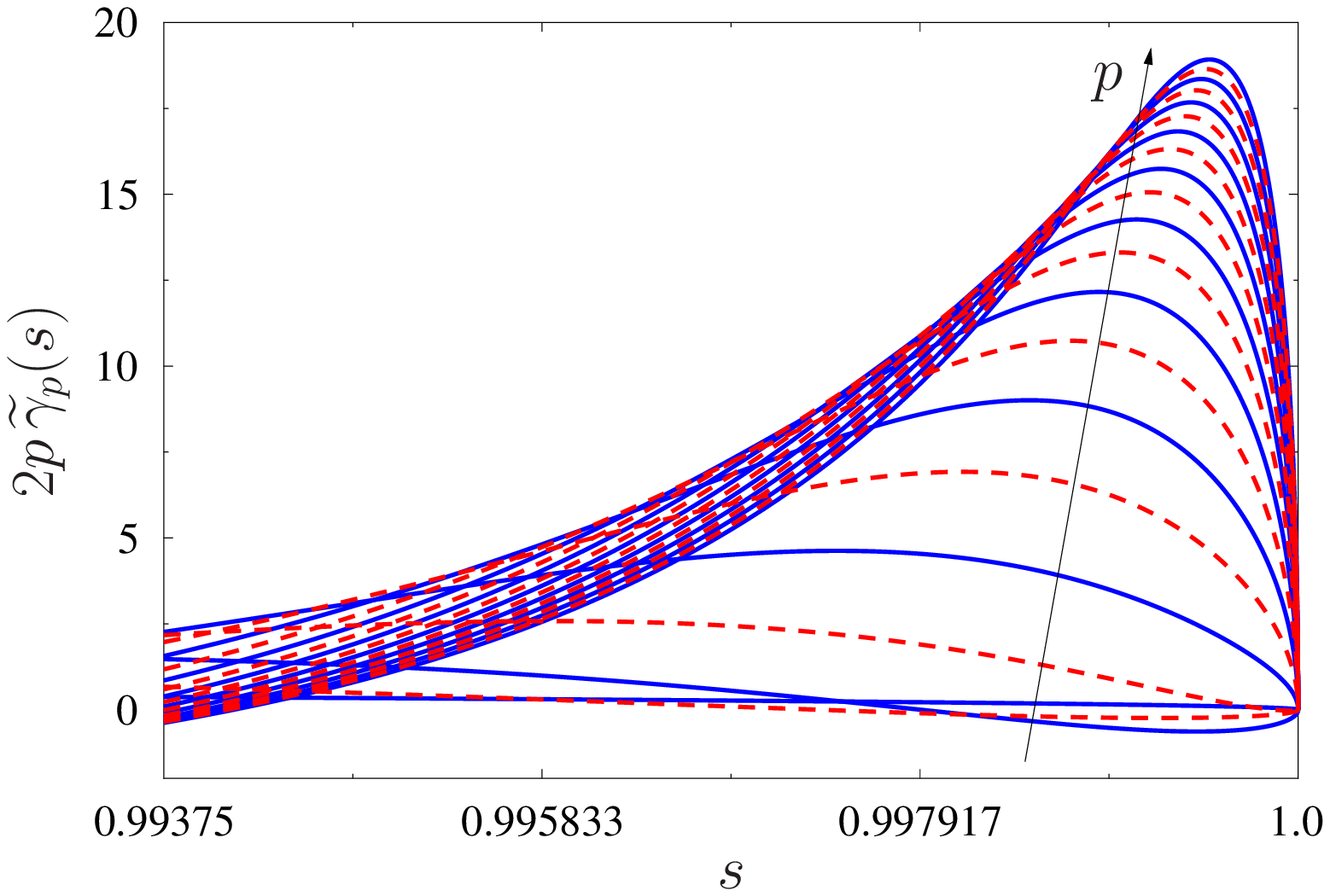}
\caption{Rescaled normalized circulation densities $2p
  \,\widetilde{\gamma}_p(s)$ obtained for $p=20,119,\dots,2000$ in a
  small neighborhood of the right endpoint $s = 1$. Red dashed lines
  and blue solid lines correspond to, respectively, $p$ even and odd.
  The trend with the increase of $p$ is indicated by an arrow.}
\label{fig:gammapZ}
\end{figure}

\section{Discussion and Conclusions}
\label{sec:final}

In this investigation we have considered the question of finding
relative equilibria involving vortex sheets in 2D potential flows. By
framing this question in terms of the Riemann-Hilbert problem in
complex analysis we proposed a general two-step approach to finding
such equilibria, where one first constructs a geometric configuration
of vortex sheets which satisfies the compatibility conditions
\eqref{eq:Cn} and then finds the corresponding circulation densities
by evaluating formula \eqref{eq:gamma}. The compatibility conditions
ensure that the circulation densities vanish at all sheet endpoints,
which is necessary for the velocity field to be everywhere well
behaved, {and additional symmetry arguments are used to ensure
  the circulation densities are real-valued}.  Using this approach
together with analytical computations we constructed a family of
rotating equilibria involving an even number $2p$ of radially-oriented
straight segments revolving about a common center of rotation where
they meet, cf.~figure \ref{fig:sheets}, which generalizes the
well-known solution \eqref{eq:sp1} corresponding to a single rotating
vortex sheet. The circulation densities of the rotating sheets are
given in closed form for the integral formula \eqref{eq:gammapp}, the
evaluation of which however requires a careful numerical approach
described in Section \ref{sec:numer}. This construction was
facilitated by a judicious choice of the geometric configuration $L$,
cf.~\eqref{eq:Lm}, consisting of $p$ bent sheets with endpoints at the
vertices of a regular $2p$-gon whose midpoints touch at the center of
rotation. As a result, it was not necessary to consider the center of
rotation as a separate sheet endpoint, although the obtained
circulation densities $\gamma_p(s)$ do vanish there (i.e., at $s = 0$)
for all $p > 1$, cf.~\eqref{eq:gammapp} and figure
\ref{fig:gammap}(a), meaning that the center of rotation may in fact
be interpreted as a sheet endpoint when $p > 1$. {It is possible
  that these new equilibria could also be obtained using other
  methods, such as, e.g., conformal transformations based on the
  Schwarz-Christoffel map.}

We found that in the limit of the equilibrium configurations
consisting of a large number of sheets, i.e., when $p \rightarrow
\infty$, the normalized circulation densities
$\widetilde{\gamma}_p(s)$, cf.~\eqref{eq:tildegammap}, vanish
everywhere except for a vanishing neighborhood of the outer endpoint
where the entire circulation of the sheet is localized in the form of
a rapid oscillation with intensity increasing with $p$. Thus, the
limiting configuration appears to have the form of a hollow vortex of
unit radius bounded by a constant-density circular vortex sheet, which
is also a well known equilibrium solution. It must be, however,
emphasized that this limit ought to be considered in some suitable
``weak'' sense, rather than in the classical pointwise (uniform)
sense. The reason is that in the neighborhood of $s = 1$ the values of
$2p \,\widetilde{\gamma}_p(s)$ diverge as $p \rightarrow \infty$ with
the sequence approaching a singular distribution at $s = 1$,
cf.~figure \ref{fig:gammapZ}. This limiting distribution clearly does
not belong to the H\"older space in which the circulation densities
most be defined \cite{Muskhelishvili2008} and is also not consistent
with conditions \eqref{eq:gamma0}. {Since both the relative equilibria
  constructed here and the hollow vortex with a vortex sheet are
  solutions of 2D Euler equations defined in a suitable distributional
  sense \cite{mb02}, we thus have a countably infinite sequence of
  distributional solutions of the first type converging to the
  solution of the second type which is characterized by a different
  topology. Understanding the precise mathematical sense of this
  convergence is an interesting open question.}

Since our ansatz for the rotating equilibrium configurations involves
an even number ($2p$) of segments, cf.~figure \ref{fig:sheets}, an
interesting question concerns construction of equilibria consisting of
an odd number of segments (such as, e.g., $\upY$). In such case, the
center of rotation would need to be treated as a separate endpoint
resulting in more complicated form of the compatibility conditions
\eqref{eq:Cn}. However, we expect that in the limit of a large number
of segments the behavior of such hypothetical equilibria would be
similar to what we observed here in figures \ref{fig:gammap}(b) and
\ref{fig:gammapZ}. {A related interesting question concerns the
  existence of equilibria with less symmetry. In this context, we have
  considered the following configuration of three sheets
\begin{equation}
L = [-1,-\alpha] \ \cup \ [\alpha,1] \ \cup \ [-i\beta,i\beta]
\label{eq:L3}
\end{equation}
defined for some $\alpha, \beta > 0$ and observed that for all $\alpha
\in (0,1)$ there exists $\beta = \beta(\alpha)$ such that this
configuration forms an equilibrium. Moreover, $\beta(\alpha)
\rightarrow 1$ as $\alpha \rightarrow 0$, meaning that in this limit
configuration \eqref{eq:L3} approaches the equilibrium corresponding
to $p=2$, cf.~\eqref{eq:gamma2ex} and figure \ref{fig:sheets}(b).
Discussion of this and other equilibria generalizing the
configurations found in this study and possibly involving curved
vortex sheets is deferred to a future publication.}

Another, more fundamental, question concerns equilibria in a
translating frame of reference. Using $f(z) = U$ in \eqref{eq:Cn} it
can be shown that for the geometric configuration \eqref{eq:Lm}
already the first compatibility condition is violated, i.e., $C_0 \neq
0$, demonstrating that this configuration cannot form a translating
equilibrium. While in the present study configuration \eqref{eq:Lm}
was proposed based on evidence from numerical experiments, in more
general situations geometric configurations $L$ which form equilibria
can be found by solving equations \eqref{eq:Cn} treated as a problem
defining $L$.  Yet another interesting problem is the stability of the
equilibrium configurations discovered here. We plan to address some of
these questions in the near future.

\section*{Acknowledgments}

We thank Marcel Rodney for helpful discussions at early stages of this
research project in 2011 {and Kevin O'Neil for bringing
  configuration \eqref{eq:L3} to our attention}. The first author
acknowledges partial support through an NSERC (Canada) Discovery
Grant. The second author was partially supported by the JST Kakenhi
(B) grant No.~18H01136, the RIKEN iTHEMS program and a grant from the
Simons Foundation. {The authors would also like to thank the
  Isaac Newton Institute for Mathematical Sciences for support and
  hospitality during the programme ``Complex analysis: techniques,
  applications and computations'' where the final version of this
  paper was prepared. This work was supported by EPSRC grant number
  EP/R014604/1.}

\FloatBarrier

\appendix
\section{Evaluation of Singular Integrals \eqref{eq:PsiPhi}}
\label{sec:PsiPhi}

For modest values of $k$ the principal-value integrals
\eqref{eq:Psik}--\eqref{eq:Phik} can be evaluated symbolically, for
example, using the software package {\tt Maple}. Denoting $X :=
\cos\phi$ and $Y :=
\frac{\arccoth\left(\frac{1}{\sin\phi}\right)}{\sin\phi}$, where $\phi
\in [0,\pi/2]$, these integrals take the form
\begin{subequations}
\label{eq:Psi0246}
\begin{align}
\Psi_0(\cos\eta) & = Y,  \\
\Psi_2(\cos\eta) & = (2X^2-1)Y + 2, \\
\Psi_4(\cos\eta) & = (8X^4-8X^2+1)Y+8X^2-\frac{8}{3}, \\
\Psi_6(\cos\eta) & = (32X^6-48X^4+18X^2-1)Y+32X^4-\frac{80 X^2}{3}+\frac{46}{15}, \\
\dots & \nonumber
\end{align}
\end{subequations}
\begin{subequations}
\label{eq:Phi02468}
\begin{align}
\Phi_0(\cos\eta) & = \frac{\pi}{2},  \\
\Phi_2(\cos\eta) & = \pi X^2, \\
\Phi_4(\cos\eta) & = \pi (4 X^4-2 X^2), \\
\Phi_6(\cos\eta) & = \pi (16X^6-16X^4+3X^2), \\
\Phi_8(\cos\eta) & = \pi (64X^8-96X^6+40X^4-4X^2), \phantom{==========} \\
\dots & \nonumber
\end{align}
\end{subequations}
where due to the properties of the function $f(\theta)$ we considered
only even values of $k$. Evidently, the expressions in
\eqref{eq:Psi0246} can be represented as $\Psi_{2m}(\phi) =
p_{2m}(X)\, Y + r_{2m}(X)$ and as can be verified by inspection the
polynomials $p_{2m}(X)$ and $r_{2m}(X)$ satisfy the following
recursion relations
\begin{subequations}
\label{eq:recPsi}
\begin{align}
p_{2m+2}(X) & = (4X^2-2)p_{2m}(X)-p_{2m-2}(X), \quad m \geqq 1, \\ 
p_0(X) &=1, \quad p_2(X)=2X^2-1, \\
r_{2m+2}(X) & = q_{2m+2}(X) + s_{2m+2}, \\
q_{2m+2}(X) & = (4X^2-2)r_{2m}(X)-r_{2m-2}(X), \quad m \geqq 1, \\ 
q_0(X) &=0, \quad q_2(X)=2, \\
s_{2m+2} &= \frac{N_{2m+2}}{D_{2m+2}}, \quad N_{2m} = (-1)^m \, 4, \\
D_{2m+2} & = 2 D_{2m} - D_{2m-2} + 8, \quad m \geqq 1, \quad D_0 = 3, \ D_2 = 1,
\end{align}
\end{subequations}
where $s_{2m}$ are rational numbers whereas $D_{2m}$ and $N_{2m}$ are
integers. Analogously, the expressions in \eqref{eq:Phi02468} can be
represented as $\Phi_{2m}(\phi) = \pi Q_{2m}(X)$ where the polynomials
$Q_{2m}(X)$ follow the recursion relations
\begin{subequations}
\label{eq:recPhi}
\begin{align}
Q_{2m+2}(X) & = (4X^2-2)Q_{2m}(X)-Q_{2m-2}(X), \quad m \geqq 2, \\ 
Q_2(X) &= X^2, \quad Q_4(X)=4X^4 - 2X^2, \\
Q_0(X) &= \frac{1}{2}.
\end{align}
\end{subequations}
Finally, expressions \eqref{eq:Psi0246}--\eqref{eq:Phi02468} together
with the recursion relations \eqref{eq:recPsi}--\eqref{eq:recPhi} make
it possible to evaluate the principal-value integrals
\eqref{eq:Psik}--\eqref{eq:Phik} for arbitrary values of $k$.

%\section*{References}

%\bibliographystyle{plain}
%\bibliography{../Bib/vortex,../Bib/fundam,../Bib/allPROTAS,../Bib/blowup}

\begin{thebibliography}{10}

\bibitem{af11}
Mark~J. Ablowitz and Athanassios~S. Fokas.
\newblock {\em Complex Variables: Introduction and Applications}.
\newblock Cambridge University Press, {Second} edition, 2011.

\bibitem{Alben2009}
Silas Alben.
\newblock Simulating the dynamics of flexible bodies and vortex sheets.
\newblock {\em Journal of Computational Physics}, 228(7):2587 -- 2603, 2009.

\bibitem{VortexCrystals2003}
H.~Aref, P.~K. Newton, M.~A. Stremler, T.~Tokieda, and D.~L. Vainchtein.
\newblock Vortex crystals.
\newblock {\em Advances in Applied Mechanics}, 39:1--79, 2003.

\bibitem{Aref2007}
Hassan Aref.
\newblock Point vortex dynamics: A classical mathematics playground.
\newblock {\em Journal of Mathematical Physics}, 48(6):065401, 2007.

\bibitem{bss76}
G.~R. Baker, P.~G. Saffman, and J.~S. Sheffield.
\newblock Structure of a linear array of hollow vortices of finite
  cross-section.
\newblock {\em Journal of Fluid Mechanics}, 74(3):469--476, 1976.

\bibitem{bbn10a}
Andrea Barreiro, Jared Bronski, and Paul~K. Newton.
\newblock Spectral gradient flow and equilibrium configurations of point
  vortices.
\newblock {\em Proceedings of the Royal Society A: Mathematical, Physical and
  Engineering Sciences}, 466(2118):1687--1702, 2010.

\bibitem{fundam:batchelor}
G.~K. Batchelor.
\newblock {\em An Introduction to Fluid Dynamics}.
\newblock Cambridge University Press, Cambridge, 1967.

\bibitem{c89a}
R.~E. Caflisch, editor.
\newblock {\em Mathematical Aspects of Vortex Dynamics}. SIAM, 1989.

\bibitem{Clarkson2009}
P.~A. Clarkson.
\newblock Vortices and polynomials.
\newblock {\em Studies in Applied Mathematics}, 123(1):37--62, 2009.

\bibitem{d85}
D.~G. Dritschel.
\newblock The stability and energetics of corotating uniform vortices.
\newblock {\em J. Fluid Mech.}, 157:95--134, 1985.

\bibitem{gipz10a}
Federico Gallizio, Angelo Iollo, Bartosz Protas, and Luca Zannetti.
\newblock On continuation of inviscid vortex patches.
\newblock {\em Physica D: Nonlinear Phenomena}, 239(3):190 -- 201, 2010.

\bibitem{hsw07}
T.Y. Hou, V.G. Stredie, and T.Y. Wu.
\newblock Mathematical modeling and simulation of aquatic and aerial animal
  locomotion.
\newblock {\em Journal of Computational Physics}, 225(2):1603 -- 1631, 2007.

\bibitem{Jones2003}
Marvin~A. Jones.
\newblock The separated flow of an inviscid fluid around a moving flat plate.
\newblock {\em Journal of Fluid Mechanics}, 496:405--441, 2003.

\bibitem{k87}
J.~R. Kamm.
\newblock {\em Shape and stability of two--dimensional vortex regions}.
\newblock PhD thesis, Caltech, 1987.

\bibitem{fw10a}
P.~Luzzatto-Fegiz and C.~H.~K. Williamson.
\newblock Stability of conservative flows and new steady-fluid solutions from
  bifurcation diagrams exploiting variational argument.
\newblock {\em Phys. Rev. Lett.}, 104:044504, 2010.

\bibitem{mb02}
A.~J. Majda and A.~L. Bertozzi.
\newblock {\em Vorticity and Incompressible Flow}.
\newblock Cambridge University Press, 2002.

\bibitem{mst88}
D.~W. Moore, P.~G. Saffman, and S.~Tanveer.
\newblock {The calculation of some Batchelor flows: The Sadovskii vortex and
  rotational corner flow}.
\newblock {\em The Physics of Fluids}, 31(5):978--990, 1988.

\bibitem{Muskhelishvili2008}
N.~I. Muskhelishvili.
\newblock {\em Singular Integral Equations. Boundary Problems of Function
  Theory and Their Application to Mathematical Physics}.
\newblock Dover, 2nd edition edition, 2008.

\bibitem{Newton2014}
Paul~K Newton.
\newblock {Point vortex dynamics in the post-Aref era}.
\newblock {\em Fluid Dynamics Research}, 46(3):031401, may 2014.

\bibitem{Newton2007}
Paul~K Newton and George Chamoun.
\newblock Construction of point vortex equilibria via brownian ratchets.
\newblock {\em Proceedings of the Royal Society A: Mathematical, Physical and
  Engineering Sciences}, 463(2082):1525--1541, 2007.

\bibitem{ONeil2009}
Kevin~A. O'Neil.
\newblock Relative equilibria of vortex sheets.
\newblock {\em Physica D: Nonlinear Phenomena}, 238(4):379 -- 383, 2009.

\bibitem{ONeil2010}
Kevin~A. O'Neil.
\newblock Collapse and concentration of vortex sheets in two-dimensional flow.
\newblock {\em Theoretical and Computational Fluid Dynamics}, 24(1):39--44, Mar
  2010.

\bibitem{ONeil2018b}
Kevin~A. O'Neil.
\newblock Dipole and multipole flows with point vortices and vortex sheets.
\newblock {\em Regular and Chaotic Dynamics}, 23:519--529, 2018.

\bibitem{ONeil2018a}
Kevin~A. O'Neil.
\newblock Relative equilibria of point vortices and linear vortex sheets.
\newblock {\em Physics of Fluids}, 30(10):107101, 2018.

\bibitem{Pierrehumbert1980}
R.~T. Pierrehumbert.
\newblock A family of steady, translating vortex pairs with distributed
  vorticity.
\newblock {\em Journal of Fluid Mechanics}, 99(1):129--144, 1980.

\bibitem{Pocklington1895}
H.~C. Pocklington.
\newblock The configuration of a pair of equal and opposite hollow straight
  vortices of finite cross-section, moving steadily through fluid.
\newblock {\em Proc. Camb. Phil. Soc.}, 8:178--187, 1895.

\bibitem{s71}
V.~S. Sadovskii.
\newblock {Vortex regions in a potential stream with a jump of Bernoulli's
  constant at the boundary}.
\newblock {\em Appl. Math. Mech.}, 35:729, 1971.

\bibitem{fundam:saffman1}
P.~G. Saffman.
\newblock {\em Vortex Dynamics}.
\newblock Cambridge Monographs on Mechanics and Applied Mathematics. Cambridge
  University Press, Cambridge, 1992.

\bibitem{ss80}
P.~G. Saffman and R.~Szeto.
\newblock Equilibrium shapes of a pair of equal uniform vortices.
\newblock {\em Phys Fluids}, 23:2339--2342, 1980.

\bibitem{Shukla2007}
Ratnesh~K. Shukla and Jeff.~D. Eldredge.
\newblock An inviscid model for vortex shedding from a deforming body.
\newblock {\em Theoretical and Computational Fluid Dynamics}, 21(5):343, Jul
  2007.

\bibitem{tz11}
H.~Telib and L.~Zannetti.
\newblock Hollow wakes past arbitrarily shaped obstacles.
\newblock {\em Journal of Fluid Mechanics}, 669:214--224, 2011.

\bibitem{Wan1988b}
Yieh-Hei Wan.
\newblock Desingularizations of systems of point vortices.
\newblock {\em Physica D: Nonlinear Phenomena}, 32(2):277 -- 295, 1988.

\end{thebibliography}

\end{document}